\documentclass{article}
\usepackage{PRIMEarxiv}
\usepackage{lmodern}
\usepackage{tikz}
\usetikzlibrary{calc}

\usepackage{booktabs}
\usepackage{tabularx}
\usepackage{array}
\usepackage{ragged2e}

\usepackage[T1]{fontenc}
\usepackage[utf8]{inputenc}
\usepackage{lmodern}
\usepackage{amsmath,amssymb,amsthm}
\usepackage[hidelinks]{hyperref}
\newcommand{\michael}[1]{#1}

\title{Maximum Solow--Polasky Diversity Subset Selection Is NP-hard Even in the Euclidean Plane}
\author{%
Michael T.~M. Emmerich$^{1}$\thanks{\href{https://orcid.org/0000-0002-7342-2090}{ORCID: 0000-0002-7342-2090}}
\and
Ksenia Pereverdieva$^{2}$\thanks{\href{https://orcid.org/0000-0002-9111-8359}{ORCID: 0000-0002-9111-8359}}
\and
Andr\'e H.~Deutz$^{2}$\thanks{\href{https://orcid.org/0000-0002-9047-6533}{ORCID: 0000-0002-9047-6533}}\\[1ex]
\small $^{1}$Faculty of Information Technology, University of Jyv\"askyl\"a, Finland\\
\small $^{2}$Leiden Center of Advanced Computer Science, Leiden University, The Netherlands
}
\date{}

\newtheorem{theorem}{Theorem}
\newtheorem{lemma}{Lemma}
\newtheorem{proposition}{Proposition}
\newtheorem{definition}{Definition}
\newtheorem{remark}{Remark}
\newtheorem{corollary}{Corollary}

\begin{document}
\maketitle

\begin{abstract}
We prove that, for every fixed $\theta_0>0$, selecting a subset of prescribed cardinality that maximizes the Solow--Polasky diversity indicator is NP-hard for finite point sets in $\mathbb{R}^2$ with the Euclidean metric, and therefore also for finite point sets in $\mathbb{R}^d$ for every fixed dimension $d\ge 2$. This strictly strengthens our earlier NP-hardness result for general metric spaces by showing that hardness persists under the severe geometric restriction to the Euclidean plane. At the same time, the Euclidean proof technique is different from the conceptually easier earlier argument for arbitrary metric spaces, and that general metric-space construction does not directly translate to the Euclidean setting. In the earlier proof one can use an exact construction tailored to arbitrary metrics, essentially exploiting a two-distance structure. In contrast, such an exact realization is unavailable in fixed-dimensional Euclidean space, so the present reduction requires a genuinely geometric argument. Our Euclidean proof is based on two distance thresholds, which allow us to separate yes-instances from no-instances by robust inequalities rather than by the exact construction used in the general metric setting. The main technical ingredient is a bounded-box comparison lemma for the nonlinear objective $\mathbf{1}^{\top}Z^{-1}\mathbf{1}$, where $Z_{ij}=e^{-\theta_0 d(x_i,x_j)}$. This lemma controls the effect of perturbations in the pairwise distances well enough to transfer the gap created by the reduction. The reduction is from \emph{Geometric Unit-Disk Independent Set}. We present the main argument in geometric form for finite subsets of $\mathbb{R}^2$, with an appendix supplying the bit-complexity details needed for polynomial-time reducibility. The result thus shows that the difficulty of optimizing this diversity measure is not merely an artifact of unrestricted metrics, but remains present already in a structured geometric environment.
\end{abstract}

\section{Introduction}
Diversity measures based on pairwise distances play an important role in
ecology, conservation biology, and optimization, where they are used to
quantify how well a set of objects, species, or candidate solutions
covers the available variation. One influential example is the
Solow--Polasky diversity indicator, introduced by Solow and
Polasky~\cite{SolowPolasky1994}. In the axiomatic framework of
Leinster~\cite{Leinster2021} and Leinster and
Meckes~\cite{LeinsterMeckes2016}, closely related quantities arise under
the name of magnitude; see also Huntsman~\cite{Huntsman2023} for a
recent optimization-oriented discussion. In indicator-based diversity
optimization, Solow--Polasky diversity has been studied by Ulrich,
Bader, and Thiele~\cite{UlrichBaderThiele2010} and in Ulrich's thesis
work~\cite{Ulrich2012}. The subset-selection viewpoint for such
indicators was formulated explicitly by Ulrich~\cite{Ulrich2012} and
was revisited in engineering design settings in the recent comparative
study of Pereverdieva et al.~\cite{Pereverdieva2025}. Related
portfolio-selection ideas balancing quality and Solow--Polasky diversity were also
considered by Yevseyeva et al.~\cite{Yevseyeva2019} in the context of
drug discovery.

For a finite set $S=\{x_1,\dots,x_k\}$ in a metric space and a fixed
parameter $\theta_0>0$, the Solow--Polasky diversity indicator is
\[
SP_{\theta_0}(S)=\mathbf{1}^{\top}Z^{-1}\mathbf{1},
\qquad
Z_{ij}=e^{-\theta_0 d(x_i,x_j)},
\]
whenever the similarity matrix $Z$ is invertible. Thus, larger distances
correspond to smaller off-diagonal similarity entries. In particular,
making points (or species) more separated decreases the off-diagonal
entries of $Z$, which in the regime studied below increases the
Solow--Polasky value.

A recent companion preprint by the same authors establishes that
selecting a maximum Solow--Polasky diversity subset is NP-hard in
general metric spaces~\cite{EmmerichPD26}. That reduction is from
Independent Set and uses a highly structured metric with only two
non-zero distance values, together with matrix entries scaled by the
input instance. Although conceptually simpler than the present proof,
the argument in~\cite{EmmerichPD26} is not a routine
Independent-Set-to-diversity reduction. The reason is that the
Solow--Polasky objective depends nonlinearly on the distance structure
through the matrix inverse $Z^{-1}$. The key idea in
\cite{EmmerichPD26} is to control this nonlinearity by means of a
Neumann-series expansion combined with a monotonicity argument.

The present paper builds on the same underlying perspective, but shows
that maximum Solow--Polasky diversity subset selection remains NP-hard
even in the Euclidean plane. Thus, the earlier general-metric result can
be tightened substantially: NP-hardness persists even under the severe
geometric restriction to finite point sets in $\mathbb{R}^2$. At the
same time, this strengthening does not follow by a direct adaptation of
the metric-space construction. In fixed dimension, one cannot expect an
exact realization of the two-distance metric used in
\cite{EmmerichPD26}. The main new ingredient is therefore a robust
comparison argument that replaces the exact two-distance construction by
a genuinely geometric one in $\mathbb{R}^2$. For clarity, we present the
core argument in geometric form for finite subsets of $\mathbb{R}^2$;
the additional bit-complexity discussion needed to realize the
construction on rational point sets in $\mathbb{Q}^2$ and hence obtain a
polynomial-time reduction is provided separately in an appendix.

The geometric source problem is the standard point-set formulation of
independent set in unit-disk graphs.

\begin{definition}[Geometric unit-disk independent set]
Given a finite point set $P\subseteq \mathbb{R}^2$ and an integer $k$,
determine whether there exists a subset $S\subseteq P$ with $|S|=k$ such
that
\[
\|p-q\|_2>1
\qquad
\text{for all distinct } p,q\in S.
\]
Equivalently, if $G(P)$ denotes the graph on vertex set $P$ in which two
points are adjacent iff their Euclidean distance is at most $1$, then the
question is whether $G(P)$ contains an independent set of size $k$.
\end{definition}

\begin{proposition}[Known source problem]
\label{prop:known-source}
Geometric unit-disk independent set is NP-hard.
\end{proposition}

\begin{remark}
This point-set formulation is standard in the unit-disk-graph literature.
Clark, Colbourn, and Johnson~\cite{ClarkColbournJohnson1990} introduced the
proximity-model view of unit-disk graphs, and later computational-geometry
papers routinely formulate maximum independent set directly on point sets
$P\subseteq \mathbb{R}^2$ via the graph $G(P)$
\cite{DasFonsecaJallu2020,TkachenkoWang2025}.
\end{remark}

Our main result is the following.

\begin{theorem}
\label{thm:main}
For every fixed constant $\theta_0>0$, the problem of selecting a
subset of prescribed cardinality maximizing the Solow--Polasky
diversity indicator is NP-hard even for finite subsets of
$\mathbb{R}^2$ under the Euclidean metric.
\end{theorem}

The remainder of the paper is organized as follows. In
Section~2 we establish a bounded-box comparison lemma for the
Solow--Polasky objective, which provides the monotonicity property on
which the reduction rests. Section~3 records a simple geometric margin
observation for finite point sets in the plane. In Section~4 these
ingredients are combined to prove the main NP-hardness result. The
paper concludes with a brief discussion of the scope of the argument.
The appendices contain complementary material. Appendix~A supplies the
bit-complexity analysis needed to realize the construction on rational
point sets in $\mathbb{Q}^2$ and thereby obtain a polynomial-time
reduction. Appendix~B presents a few small numerical sanity checks that
illustrate the separation mechanism underlying the proof.

\section{A bounded-box comparison lemma}

The key analytic input is that, once all off-diagonal similarities are
sufficiently small, the Solow--Polasky objective decreases whenever one
of those similarities increases. Since each off-diagonal entry has the
form $e^{-\theta_0 d(x_i,x_j)}$, this means equivalently that, in this
small-similarity regime, increasing a pairwise distance increases the
objective value.

\begin{lemma}[Bounded-box comparison]
\label{lem:box}
Fix $k\ge 2$ and let $0<\rho<1/(4k)$. Let $Z$ be a symmetric
$k\times k$ matrix with diagonal entries $1$ and all off-diagonal
entries in $[0,\rho]$. Define
\[
F(Z):=\mathbf{1}^{\top}Z^{-1}\mathbf{1}.
\]
Then:
\begin{enumerate}
\item $Z$ is invertible and the vector $w:=Z^{-1}\mathbf{1}$ has
strictly positive components;
\item $F(Z)$ is strictly decreasing in each off-diagonal entry of $Z$.
\end{enumerate}
Consequently, for any $0\le r\le q\le \rho$:
\begin{enumerate}
\item if all off-diagonal entries of $Z$ are at most $r$, then
\[
F(Z)\ge \frac{k}{1+(k-1)r};
\]
\item if at least one off-diagonal entry is at least $q$, then
\[
F(Z)\le k-\frac{2q}{1+q}.
\]
\end{enumerate}
\end{lemma}

\begin{proof}
We split the proof into three steps.

\smallskip
\noindent
\emph{Step 1: invertibility and positivity.}
Write $Z=I+B$, where $B$ has zero diagonal and off-diagonal entries in
$[0,\rho]$. Every row sum of $|B|$ is at most $(k-1)\rho<1/4$, so
\[
\|B\|_\infty<1.
\]
Hence the Neumann series converges:
\[
Z^{-1}=(I+B)^{-1}=\sum_{m=0}^{\infty}(-B)^m.
\]
Applying this to $\mathbf{1}$ gives
\[
w:=Z^{-1}\mathbf{1}=\sum_{m=0}^{\infty}(-B)^m\mathbf{1}.
\]
For each component,
\[
w_i\ge 1-\sum_{m=1}^{\infty}\|B\|_\infty^m
>1-\sum_{m=1}^{\infty}(1/4)^m
=\frac23>0.
\]
So $Z$ is invertible and all components of $w$ are strictly positive.

\smallskip
\noindent
\emph{Step 2: monotonicity.}
Fix one off-diagonal variable $z_{ab}=z_{ba}$ with $a\neq b$. Using
\[
\frac{\partial Z^{-1}}{\partial z_{ab}}
=
-\,Z^{-1}\bigl(E_{ab}+E_{ba}\bigr)Z^{-1},
\]
we obtain
\[
\frac{\partial F}{\partial z_{ab}}
=
\mathbf{1}^{\top}\frac{\partial Z^{-1}}{\partial z_{ab}}\mathbf{1}
=
-2w_aw_b<0.
\]
Thus $F(Z)$ is strictly decreasing in each off-diagonal entry.

\smallskip
\noindent
\emph{Step 3: extremal configurations.}
If all off-diagonal entries are at most $r$, then by monotonicity
$F(Z)$ is minimized when all of them are as large as possible, namely
when
\[
Z=(1-r)I+rJ.
\]
A direct computation gives
\[
F(Z)=\frac{k}{1+(k-1)r}.
\]

If at least one off-diagonal entry is at least $q$, then by the same
monotonicity $F(Z)$ is maximized when exactly one symmetric pair equals
$q$ and all other off-diagonal entries are $0$. In that case $Z$
consists of one $2\times 2$ block
\[
\begin{pmatrix}
1 & q\\
q & 1
\end{pmatrix}
\]
and $k-2$ singleton blocks. Therefore
\[
F(Z)=(k-2)+\frac{2}{1+q}
=
k-\frac{2q}{1+q}.
\]
This proves the two bounds.
\end{proof}

\begin{corollary}[Strict separation]
\label{cor:gap}
Under the assumptions of Lemma~\ref{lem:box}, every matrix whose
off-diagonal entries are all at most $r$ has strictly larger objective
value than every matrix with at least one off-diagonal entry at least
$q$, provided
\[
q>k(k-1)r.
\]
\end{corollary}

\begin{proof}
By Lemma~\ref{lem:box}, it is enough to prove that
\[
\frac{k}{1+(k-1)r}
>
k-\frac{2q}{1+q}.
\]
We now transform this inequality step by step.

First rewrite the right-hand side over the denominator $1+q$:
\[
k-\frac{2q}{1+q}
=
\frac{k(1+q)-2q}{1+q}
=
\frac{k+(k-2)q}{1+q}.
\]
Hence the desired inequality is equivalent to
\[
\frac{k}{1+(k-1)r}
>
\frac{k+(k-2)q}{1+q}.
\]

Since both denominators are positive, we may cross-multiply:
\[
k(1+q)
>
\bigl(1+(k-1)r\bigr)\bigl(k+(k-2)q\bigr).
\]
Expanding the right-hand side gives
\[
k+kq
>
k+(k-2)q+(k-1)r\bigl(k+(k-2)q\bigr).
\]
Subtracting $k$ from both sides yields
\[
kq
>
(k-2)q+(k-1)r\bigl(k+(k-2)q\bigr).
\]
Now subtract $(k-2)q$ from both sides:
\[
2q
>
(k-1)r\bigl(k+(k-2)q\bigr).
\]
So it remains to verify this last inequality.

Because $q\le \rho<1/(4k)<1$, we have
\[
k+(k-2)q < k+(k-2)=2k-2<2k.
\]
Therefore
\[
(k-1)r\bigl(k+(k-2)q\bigr)
<
2k(k-1)r.
\]
Thus a sufficient condition for
\[
2q>(k-1)r\bigl(k+(k-2)q\bigr)
\]
is
\[
2q>2k(k-1)r,
\]
that is,
\[
q>k(k-1)r.
\]
This is exactly the assumed hypothesis. Hence
\[
\frac{k}{1+(k-1)r}
>
k-\frac{2q}{1+q},
\]
and the strict separation follows.
\end{proof}

\section{A geometric margin for finite point sets}

The next observation is purely geometric: for a finite point set,
there is always a positive gap both below the smallest nonzero distance
and above the threshold $1$.

\begin{lemma}[Finite-set margin]
\label{lem:finite-gap}
Let $P\subseteq \mathbb{R}^2$ be finite. Then there exist constants
$\delta>0$ and $\eta>0$ such that:
\begin{enumerate}
\item for all distinct $p,q\in P$, one has
\[
\|p-q\|_2\ge \delta;
\]
\item whenever $\|p-q\|_2>1$, one in fact has
\[
\|p-q\|_2\ge 1+\eta.
\]
\end{enumerate}
\end{lemma}

\begin{proof}
Because $P$ is finite, the set
\[
D:=\{\|p-q\|_2 : p,q\in P,\ p\neq q\}
\]
is finite and contained in $(0,\infty)$. Hence
\[
\delta:=\min D>0.
\]

Likewise, the set
\[
E:=\{\|p-q\|_2-1 : p,q\in P,\ \|p-q\|_2>1\}
\]
is finite and contained in $(0,\infty)$ whenever it is nonempty. If
$E\neq\emptyset$, define $\eta:=\min E$; otherwise set $\eta:=1$.
Then $\eta>0$, and every pair with distance greater than $1$ has
distance at least $1+\eta$.
\end{proof}

\section{Reduction to the Euclidean plane}

We now present the geometric core of the reduction. The proof in this
section only uses that the source point set is finite. The fact that the
construction can be carried out in polynomial time from a rationally
encoded input is deferred to Appendix~\ref{app:rational}.

\begin{proof}[Proof of Theorem~\ref{thm:main}]
Let $(P,k)$ be an instance of geometric unit-disk independent set, where
\[
P=\{p_1,\dots,p_n\}\subseteq \mathbb{R}^2.
\]
We may assume $k\ge 2$, since the cases $k=0,1$ are trivial.

By Lemma~\ref{lem:finite-gap}, there exist $\delta>0$ and $\eta>0$ such
that:
\begin{itemize}
\item every distinct pair in $P$ is at distance at least $\delta$;
\item every pair in $P$ whose distance is greater than $1$ is in fact at
distance at least $1+\eta$.
\end{itemize}

Choose a scale factor $L>0$ so large that
\[
e^{-\theta_0L\delta}\le \frac1{4k}
\qquad\text{and}\qquad
e^{-\theta_0L\eta}<\frac1{k(k-1)}.
\]
To ensure the above equation
it is enough to impose corresponding lower bounds on \(L\).
Indeed,
\[
e^{-\theta_0L\delta}\le \frac{1}{4k}
\iff
-\theta_0L\delta \le \ln\!\left(\frac{1}{4k}\right)
=
-\ln(4k)
\iff
L\ge \frac{\ln(4k)}{\theta_0\delta}.
\]
Similarly,
\[
e^{-\theta_0L\eta}< \frac{1}{k(k-1)}
\iff
-\theta_0L\eta < \ln\!\left(\frac{1}{k(k-1)}\right)
=
-\ln\!\bigl(k(k-1)\bigr)
\iff
L> \frac{\ln\!\bigl(k(k-1)\bigr)}{\theta_0\eta}.
\]
Therefore, it suffices to choose
\[
L>\max\left\{
\frac{\ln(4k)}{\theta_0\delta},
\frac{\ln\!\bigl(k(k-1)\bigr)}{\theta_0\eta}
\right\}.
\]

Now scale the point set by $L$:
\[
x_i:=Lp_i\in\mathbb{R}^2,
\qquad
X:=\{x_1,\dots,x_n\}.
\]
For any subset $U\subseteq P$, write
\[
L U:=\{Lp : p\in U\}\subseteq X.
\]

We compare two kinds of $k$-subsets of $P$.

\smallskip
\noindent
\emph{Reminder on similarities.}
Recall that larger pairwise distances produce smaller off-diagonal
entries in the similarity matrix $Z$, because the similarity is
$e^{-\theta_0 d}$. Thus independent subsets, whose mutual distances are
all larger than $1$, are expected to have smaller similarities and
hence larger Solow--Polasky values.

\smallskip
\noindent
\emph{Good subsets.}
Call a $k$-subset $S\subseteq P$ \emph{good} if it is independent, i.e.,
if all pairwise distances in $S$ are greater than $1$.

Then every pairwise distance in $S$ is at least $1+\eta$, so after
scaling every pairwise distance in $LS$ is at least $L(1+\eta)$.
Therefore every off-diagonal similarity in the similarity matrix of $LS$
is at most
\[
r:=e^{-\theta_0L(1+\eta)}.
\]
By Lemma~\ref{lem:box},
\[
SP_{\theta_0}(LS)\ge \frac{k}{1+(k-1)r}.
\]

\smallskip
\noindent
\emph{Bad subsets.}
Call a $k$-subset $T\subseteq P$ \emph{bad} if it is not independent.
Then some pair in $T$ has distance at most $1$. After scaling, the
corresponding pair in $LT$ has distance at most $L$, so the corresponding
similarity entry is at least
\[
q:=e^{-\theta_0L}.
\]
Again by Lemma~\ref{lem:box},
\[
SP_{\theta_0}(LT)\le k-\frac{2q}{1+q}.
\]

\smallskip
\noindent
\emph{Why Lemma~\ref{lem:box} applies.}
Every distinct pair in $P$ is at distance at least $\delta$, hence every
distinct pair in $X$ is at distance at least $L\delta$. So every
off-diagonal similarity in every $k$-subset of $X$ is at most
\[
\rho:=e^{-\theta_0L\delta}\le \frac1{4k}.
\]
Thus all hypotheses of Lemma~\ref{lem:box} are satisfied.

\smallskip
\noindent
\emph{Strict separation.}
Finally,
\[
\frac{q}{r}=e^{\theta_0L\eta}>k(k-1),
\]
so
\[
q>k(k-1)r.
\]
By Corollary~\ref{cor:gap}, every good $k$-subset has strictly larger
Solow--Polasky value than every bad $k$-subset.

Therefore the maximizing $k$-subsets of $X$ are exactly the sets $LS$
where $S\subseteq P$ is an independent $k$-subset. Hence an algorithm
for Solow--Polasky subset selection in the Euclidean plane would solve
geometric unit-disk independent set.

To complete the NP-hardness proof, one must verify that the scaling
factor $L$ can be chosen with polynomial encoding length when the input
point set is given by rational coordinates. This is done in
Appendix~\ref{app:rational}. Therefore the reduction is computable in
polynomial time. By Proposition~\ref{prop:known-source},
Solow--Polasky subset selection in $\mathbb{R}^2$ is NP-hard.
\end{proof}
\section{Related Work}
A related precursor is the work of Leinster and
Meckes~\cite{LeinsterMeckes2016}, who study diversity maximization over
probability distributions on a fixed symmetric similarity matrix $Z$.
Their framework is broader in one direction than ours, since it allows
arbitrary similarity matrices, and they show in particular that for the
adjacency matrix of a finite reflexive graph the maximum diversity is the independence number. Thus, they already proved NP-hardness
of diversity maximization in a broader framework. However, the
corresponding complexity question for Solow--Polasky kernels induced by metric spaces and
Euclidean metrics in fixed dimension (already for $\mathbb{R}^2$)
remained open. They also relate maximum diversity to magnitude through
the theory of weightings, showing that under suitable
hypotheses - for instance, for positive semidefinite matrices with a
nonnegative weighting, including ultrametric matrices -the maximum
diversity can be computed exactly in polynomial time.
Thus, while the results of Leinster and Meckes clearly anticipate the possibility of hardness phenomena through the connection with independence number, they do not subsume the general-metric or Euclidean subset-selection results established here.

\section{Concluding remarks}

The argument given above shows that the Euclidean-plane case does not
depend on an exact two-distance realization of the type used in the
general-metric setting. What is essential is, first, a regime in which
all off-diagonal similarities are uniformly small, so that the
Solow--Polasky objective is monotone in each similarity entry, and,
second, a positive separation between the similarities associated with
independent pairs and those associated with adjacent pairs. Once these
two structural properties have been established, the reduction from
geometric unit-disk independent set follows.

This clarifies the relation between the present paper and the earlier
general-metric result~\cite{EmmerichPD26}. The metric-space argument is
conceptually simpler because it is based on an exact two-distance
construction. The Euclidean-plane result is stronger, but it requires a
different geometric realization principle. In this sense, the present
paper should be read as a refinement of the earlier hardness proof,
rather than as a mere reformulation of it.

The appendices complete the proof in two directions. Appendix~A
supplies the quantitative estimates required to realize the construction
for rational point sets with polynomial encoding length, thereby
establishing polynomial-time reducibility in the standard Turing model.
Appendix~B contains a few small numerical examples whose role is purely
illustrative: they are not used logically in the proof, but they make
the separation mechanism visible on concrete instances.

\appendix

\section{Quantitative separation for rational inputs}
\label{app:rational}

This appendix supplies the complexity-theoretic bookkeeping omitted from
the main proof. It shows that, for rationally encoded point sets, the
parameters $\delta$, $\eta$, and the scale factor $L$ can be chosen
quantitatively with polynomial encoding length.

\begin{lemma}[Threshold margin for rational coordinates]
\label{lem:bitgap}
Let $P\subseteq \mathbb{Q}^2$ be finite, and assume every numerator and
denominator appearing in the coordinates of points of $P$ has bit-length
at most $B\ge 1$. Set
\[
\varepsilon:=2^{-12B}.
\]
Then:
\begin{enumerate}
\item for all distinct $p,q\in P$, one has $\|p-q\|_2\ge \varepsilon$;
\item if $\|p-q\|_2>1$, then in fact $\|p-q\|_2\ge 1+\varepsilon$.
\end{enumerate}
\end{lemma}

\begin{proof}
Write
\[
p=\left(\frac{a_1}{b_1},\frac{c_1}{d_1}\right),
\qquad
q=\left(\frac{a_2}{b_2},\frac{c_2}{d_2}\right),
\]
with all numerators and denominators bounded in absolute value by $2^B$.
Each coordinate difference is then a rational with denominator at most
$2^{2B}$, so each squared coordinate difference has denominator at most
$2^{4B}$. Therefore the squared distance
\[
\sigma:=\|p-q\|_2^2
\]
is a rational with denominator at most $2^{8B}$.

If $p\neq q$, then $\sigma>0$, hence $\sigma\ge 2^{-8B}$ and therefore
\[
\|p-q\|_2\ge 2^{-4B}\ge \varepsilon.
\]
This proves (1).

Now assume $\|p-q\|_2>1$, i.e. $\sigma>1$. Since $\sigma$ has
denominator at most $2^{8B}$, we have
\[
\sigma-1\ge 2^{-8B}.
\]
Also each coordinate has absolute value at most $2^B$, hence
\[
\|p-q\|_2\le 2^{B+2},
\qquad
\|p-q\|_2+1\le 2^{B+3}.
\]
Using
\[
\|p-q\|_2-1=\frac{\sigma-1}{\|p-q\|_2+1},
\]
we obtain
\[
\|p-q\|_2-1\ge 2^{-(9B+3)}\ge 2^{-12B}=\varepsilon.
\]
This proves (2).
\end{proof}

\begin{corollary}[Polynomially computable scaling]
\label{cor:poly-scale}
Let $(P,k)$ be a rationally encoded instance of geometric unit-disk
independent set, and let $B$ be the maximum coordinate bit-length.
Define
\[
\varepsilon:=2^{-12B},
\qquad
M:=\max\{4k,\;k(k-1)+1\},
\qquad
c_{\theta_0}:=\left\lceil\frac{\ln 2}{\theta_0}\right\rceil,
\]
and let
\[
L:=c_{\theta_0}\,2^{12B}\,\lceil \log_2 M\rceil.
\]
Then $L$ is an integer of polynomial encoding length, and it satisfies
\[
e^{-\theta_0L\varepsilon}\le \frac1M.
\]
In particular,
\[
e^{-\theta_0L\varepsilon}\le \frac1{4k}
\qquad\text{and}\qquad
e^{\theta_0L\varepsilon}\ge M>k(k-1).
\]
\end{corollary}

\begin{proof}
By construction,
\[
\theta_0L\varepsilon
=
\theta_0\,c_{\theta_0}\,\lceil\log_2 M\rceil
\ge
(\ln 2)\,\lceil\log_2 M\rceil
\ge
\ln M.
\]
Exponentiating yields
\[
e^{-\theta_0L\varepsilon}\le \frac1M.
\]
The bit-length of $L$ is polynomial in $B$ and $\log k$ because
$\theta_0$ is fixed and $L$ is the product of a constant depending only
on $\theta_0$, a power of $2$ of exponent $12B$, and
$\lceil\log_2 M\rceil=O(\log k)$.
\end{proof}

\begin{remark}
In the notation of the main proof, Lemma~\ref{lem:bitgap} shows that one
may take
\[
\delta\ge \varepsilon
\qquad\text{and}\qquad
\eta\ge \varepsilon.
\]
Therefore the single choice of $L$ in
Corollary~\ref{cor:poly-scale} simultaneously guarantees both
requirements used in the proof of Theorem~\ref{thm:main}:
\[
e^{-\theta_0L\delta}\le \frac1{4k}
\qquad\text{and}\qquad
e^{-\theta_0L\eta}<\frac1{k(k-1)}.
\]
This is exactly what is needed to turn the geometric argument from the
main text into a polynomial-time reduction.
\end{remark}

\section{Numerical example}

As a small numerical sanity check, we evaluate the
Solow--Polasky objective directly for a few small Euclidean point sets
using
\[
SP_{\theta_0}(S)=\mathbf{1}^{\top}Z^{-1}\mathbf{1},
\qquad
Z_{ij}=e^{-\theta_0\|x_i-x_j\|_2},
\]
with $\theta_0=1$ and scale factor $L=3$.
The purpose of this appendix is not to furnish a proof, but to make the
mechanism of the reduction more explicit on a concrete instance.

\medskip
\noindent
\textbf{A worked example.}
Consider the geometric unit-disk instance
\[
P=\{p_1,p_2,p_3\}=\left\{\left(0,0\right),\left(1,0\right),\left(0,\tfrac34\right)\right\},
\qquad k=2.
\]

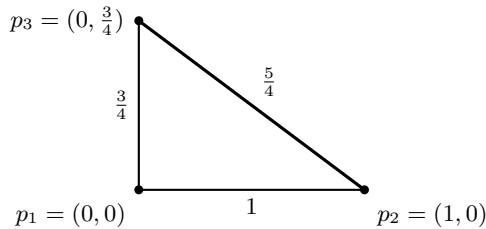
\begin{figure}[h!]
\centering
\begin{tikzpicture}[scale=3.0,
    every node/.style={font=\small},
    pt/.style={circle, fill=black, inner sep=1.2pt}
]

\coordinate (p1) at (0,0);
\coordinate (p2) at (1,0);
\coordinate (p3) at (0,0.75);

\draw[thick] (p1) -- (p2);
\draw[thick] (p1) -- (p3);

\draw[very thick] (p2) -- (p3);

\node[pt,label=below left:{$p_1=(0,0)$}] at (p1) {};
\node[pt,label=below right:{$p_2=(1,0)$}] at (p2) {};
\node[pt,label=left:{$p_3=(0,\tfrac34)$}] at (p3) {};

\node[below] at ($(p1)!0.5!(p2)$) {$1$};
\node[left]  at ($(p1)!0.5!(p3)$) {$\tfrac34$};
\node[above right] at ($(p2)!0.5!(p3)$) {$\tfrac54$};

\end{tikzpicture}
\caption{The three-point instance used in Appendix~B. The highlighted edge
corresponds to the unique independent pair.}
\label{fig:triangle-example}
\end{figure}
These three points form a nondegenerate triangle with one short side,
one side exactly at the threshold, and one long side. More precisely,
\[
\|p_1-p_3\|_2=\tfrac34<1,
\qquad
\|p_1-p_2\|_2=1,
\qquad
\|p_2-p_3\|_2=\sqrt{1+\tfrac{9}{16}}=\tfrac54>1.
\]
Hence the unique independent $2$-subset in the source instance is
$\{p_2,p_3\}$.

The reduction scales the point set by $L=3$ and produces
\[
X=LP=\left\{x_1,x_2,x_3\right\}=\left\{\left(0,0\right),\left(3,0\right),\left(0,\tfrac94\right)\right\}.
\]
The corresponding pairwise distances in the image are
\[
\|x_1-x_3\|_2=\tfrac94,
\qquad
\|x_1-x_2\|_2=3,
\qquad
\|x_2-x_3\|_2=\tfrac{15}{4}.
\]
For a two-point set whose distance is $d$, the similarity matrix is
\[
Z(d)=\begin{pmatrix}1 & e^{-d}\\ e^{-d} & 1\end{pmatrix},
\]
and a direct calculation gives
\[
Z(d)^{-1}=\frac{1}{1-e^{-2d}}
\begin{pmatrix}1 & -e^{-d}\\ -e^{-d} & 1\end{pmatrix},
\qquad
SP_{1}(d)=\mathbf{1}^{\top}Z(d)^{-1}\mathbf{1}=\frac{2}{1+e^{-d}}.
\]
Thus the three feasible $2$-subsets of $X$ have values
\[
SP_{1}(\{x_1,x_3\})=\frac{2}{1+e^{-9/4}}\approx 1.809301,
\]
\[
SP_{1}(\{x_1,x_2\})=\frac{2}{1+e^{-3}}\approx 1.905148,
\]
and
\[
SP_{1}(\{x_2,x_3\})=\frac{2}{1+e^{-15/4}}\approx 1.954045.
\]
Since $SP_{1}(d)$ is increasing in the distance $d$, the maximizing pair
is the one corresponding to the long side of the triangle, namely
$\{x_2,x_3\}$. This is exactly the image under the reduction of the
unique independent $2$-subset $\{p_2,p_3\}$ of the source instance.
Thus, on this concrete example, the reduction selects the long side of
the triangle and excludes the short side and the threshold side, in
complete agreement with the proof of Theorem~\ref{thm:main}.

The worked example already illustrates the mechanism relevant for the
proof: after scaling, the independent pair corresponds to the largest
distance and therefore to the smallest off-diagonal similarity, which in
turn yields the largest value of the Solow--Polasky objective. For the
purposes of the present paper, this explicit three-point calculation is
more informative than a collection of further numerical examples and
suffices as a concrete sanity check for the reduction.

\end{document}